\documentclass[letterpaper,11pt]{article} 
\usepackage[hmargin=1in,vmargin=1in]{geometry}
\usepackage{amsmath, amssymb}
\usepackage{bbm}
\usepackage{graphicx}
\usepackage{algorithm}
\usepackage{algorithmic}
\usepackage{xspace}
\title{\vspace{-1cm}Approximating Graphic TSP by Matchings\thanks{This research
was supported by ERC Advanced investigator grant 226203.}}
\author{Tobias M\"{o}mke and Ola Svensson \\
 Royal Institute of Technology - KTH, Stockholm, Sweden \\
  {\tt \{moemke,osven\}@kth.se }}

\newtheorem{theorem}{Theorem}[section]

\newtheorem{lemma}[theorem]{Lemma}
\newtheorem{claim}[theorem]{Claim}

\newtheorem{definition}[theorem]{Definition}

\newenvironment{proofclaim}{\begin{trivlist}
    \item[\hskip\labelsep {\it Proof of Claim}.]}{\QED \end{trivlist}}
\newenvironment{proof}{\begin{trivlist}
    \item[\hskip\labelsep {\bf Proof}.]}{\QED \end{trivlist}}

\newenvironment{proofof}[1]{\begin{trivlist}
    \item[\hskip\labelsep {\bf Proof of #1.}]}{\QED \end{trivlist}}

\newenvironment{applemma}{\begin{trivlist}
        \item[\hskip\labelsep {\bf Lemma}]}{ \end{trivlist}}
\newenvironment{apptheorem}{\begin{trivlist}
        \item[\hskip\labelsep {\bf Theorem}]}{ \end{trivlist}}

\newcommand{\QED}{\hfill $\square$}

\newcommand{\hide}[1]{
}

\newcounter{myclaim}
\newcommand{\MS}{\ensuremath{\mbox{removable pairing}}\xspace}
\newcommand{\LP}[1]{\ensuremath{LP(#1)}}
\newcommand{\OLP}[1]{\ensuremath{OPT_{LP}(#1)}}

\newcommand{\oa}[1]{\ensuremath{\overrightarrow{#1}}}

\newcommand{\TSP}{graph-TSP\xspace}
\newcommand{\HPP}{graph-TSPP\xspace}

\newcommand{\dist}[1]{\ensuremath{dist(#1)}}

\begin{document}
\maketitle

\begin{abstract}
  We present a framework for approximating the metric TSP based on a
  novel use of matchings. Traditionally, matchings have been used to
  add edges in order to make a given graph Eulerian, whereas our
  approach also allows for the removal of certain edges leading to a
  decreased cost.

  For the TSP on graphic metrics (\TSP), the approach yields a
  $1.461$-approximation algorithm with respect to the
  Held-Karp lower bound. For \TSP restricted to a class of graphs that
  contains degree three bounded and claw-free graphs, we show that the
  integrality gap of the Held-Karp relaxation matches the conjectured
  ratio $4/3$. The framework allows for generalizations in a natural way and
  also leads to a $1.586$-approximation algorithm for the traveling
  salesman path problem on graphic metrics where the start and end vertices are
  prespecified.
\end{abstract}

\section{Introduction}\label{sec:intro}
    The traveling salesman problem in metric graphs is one of most
    fundamental NP-hard optimization problems. In spite of a vast amount
    of research several important questions remain open. While the problem
    is known to be APX-hard and NP-hard to approximate with a ratio better than
    $220/219$ \cite{PV06}, the best upper bound is still the
    1.5-approximation algorithm obtained by Christofides~\cite{Chr76}
    more than three decades ago.
    A promising direction to improve this approximation guarantee, has
    long been to understand the power of a linear program known as the
    Held-Karp relaxation~\cite{HK70}. On the one hand, the best lower
    bound on its integrality gap (for the symmetric case) is $4/3$ and
    indeed conjectured to be tight \cite{Goe95}. On the other hand, the
    best known analysis~\cite{SW90, Wol80} is based on Christofides'
    algorithm and gives an upper bound on the integrality gap of $1.5$.

    In the light of this difficulty of even determining the integrality gap of
    the Held-Karp relaxation, a reasonable way to approach the metric TSP
    is to restrict the set of feasible inputs. One promising candidate is
    the \emph{\TSP,} that is, the traveling salesman problem where
    distances between cities are given by any graphic metric, i.\,e., the
    distance between two cities is the length of the shortest path in a
    given (unweighted) graph. Equivalently, \TSP can be formulated as the
    problem of finding an Eulerian multigraph within an unweighted input
    graph so as to minimize the number of edges. In contrast to TSP on
    Euclidean metrics that admits a PTAS~\cite{Arora98,Mitch99}, the \TSP
    seems to capture the difficulty of the metric TSP in the sense that,
    as stated in \cite{GKP95}, it is APX-hard and the lower bound 4/3 on
    the integrality gap of the Held-Karp relaxation is established using a
    \TSP instance.

    The TSP on graphic metrics has recently drawn considerable attention.
    In 2005, Gamarnik et al.~\cite{GLS05} showed that for cubic
    3-edge-connected graphs, there is an approximation algorithm achieving
    an approximation ratio of $1.5-5/389$. This result was generalized to
    cubic graphs by Boyd et al.~\cite{BSSS11}, who obtained an improved
    performance guarantee of 4/3.  For subcubic graphs, i.\,e., graphs of
    degree at most $3$, they also gave an 7/5-approximation algorithm with
    respect to the Held-Karp lower bound.  In a major achievement, Gharan
    et al.~\cite{GSS11} recently presented an approximation algorithm for
    \TSP with performance guarantee strictly better than 1.5. The approach
    in~\cite{GSS11} is similar to that of Christofides in the sense that
    they start with a spanning tree and then add a perfect matching of
    those vertices of odd-degree to make the graph Eulerian. The main
    difference is that instead of starting with a minimum spanning tree,
    their approach uses the solution of the Held-Karp relaxation to sample
    a spanning tree. Although the proposed algorithm in~\cite{GSS11} is
    surprisingly simple, the analysis is technically involved and several novel
    ideas are needed to obtain the improved performance guarantee
    $1.5-\epsilon$ for an $\epsilon$ of the order $10^{-12}$.

\paragraph{Our Results and Overview of Techniques.}
    We propose an alternative framework for approximating the metric
    TSP and use it to obtain  an improved approximation
    algorithm for \TSP.
    \begin{theorem}\label{thm:approximationratio}
      There is a polynomial time approximation algorithm for \TSP with
      performance guarantee $\frac{14\cdot ( \sqrt{2}-1)}{12\cdot
        \sqrt{2}-13} < 1.461$.
    \end{theorem}
    The result implies an upper bound on the integrality gap of the Held-Karp
    relaxation for \TSP that matches the approximation ratio.
    For the restricted class of graphs, where each block (i.\,e., each
    maximally 2-vertex-connected subgraph) is either claw-free or of
    degree at most $3$, we use the framework to construct a polynomial
    time $4/3$-approximation algorithm showing that the conjectured
    integrality gap of the Held-Karp relaxation is tight for those
    graphs. In fact, the techniques allow us to prove the tight result
    that any $2$-vertex-connected graph of degree at most $3$ has a
    spanning Eulerian multigraph with at most $4n/3 - 2/3$ edges, which
    settles a conjecture of Boyd et al.~\cite{BSSS11} affirmatively.

    Our framework is based on earlier works by Frederickson \&
    Ja'ja'~\cite{FJJ89} and Monma et al.~\cite{MMP90}, who related the
    cost of an optimal tour to the size of a minimum $2$-vertex-connected
    subgraph.  More specifically, Monma et al. showed that a
    $2$-vertex-connected graph $G=(V,E)$ always has a spanning Eulerian
    multigraph with at most $\frac{4}{3} |E|$ edges, generalizing a
    previous result of Frederickson \& Ja'ja' who obtained
    the same result for the special case of planar $2$-vertex-connected
    graphs.
    One interpretation of their approaches is the following.  Given a
    $2$-vertex-connected graph $G=(V,E)$, they show how to pick a random
    subset $M$ of edges satisfying: (i) an edge is in $M$ with probability
    $1/3$ and (ii) the multigraph $H$ with vertex set $V$ and edge set $E
    \cup M$ is spanning and Eulerian. From property $(i)$ of
    $M$, the expected number of edges in $H$ is $\frac{4}{3} |E|$ yielding
    their result.

    Although the factor $4/3$ is asymptotically tight for some classes of
    graphs (one example is the family of integrality gap instances for the
    Held-Karp relaxation described in Section~\ref{sec:prelim}), the bound
    rapidly gets worse for $2$-vertex-connected graphs with significantly
    more than $n$ edges. The novel idea to overcome this issue is the
    following.  Instead of adding all the edges in $M$ to $G$, some of the
    edges in $M$ might instead be removed from $G$ to form $H$. As long as
    the removal of the edges does not disconnect the graph, this will
    again result in a spanning Eulerian multigraph $H$. To specify a
    subset $R$ of edges that safely may be removed we introduce, in
    Section~\ref{sec:tspframework}, the notion of a ``\MS''. The framework
    is then completed by Theorem~\ref{thm:main}, where we show that a
    $2$-vertex-connected graph $G=(V,E)$ with a set $R$ of removable edges
    has a spanning Eulerian multigraph with at most $\frac{4}{3} |E| - \frac{2}{3} |R|$
    edges.

    In order to use the framework, one of the main challenges is to find a
    sufficiently large set of removable edges. In
    Section~\ref{sec:circulation}, we show that this problem can be
    reduced to that of finding a min-cost circulation in a certain
    circulation network.  To analyze the circulation network we then (in
    Section~\ref{sec:algorithms}) use several properties of an extreme point solution to
    the Held-Karp relaxation to obtain our main algorithmic result. The
    better approximation guarantees for special graph classes follows from
    that the circulation network has an easier structure in these cases,
    which in turn allows for a better analysis.

    Finally, we note that the techniques generalize in a natural way. Our
    results can be adapted to the more general traveling salesman path
    problem (\HPP) with prespecified start and end vertices to improve on the
    approximation ratio of $5/3$ by Hoogeveen \cite{Hoo91} when
    considering graphic metrics. More specifically, we obtain the
    following.
    \begin{theorem}\label{thm:approximationratiohpp}
        For any $\varepsilon > 0$, there is a polynomial time approximation algorithm for
        \HPP with performance guarantee 
        $ 3-\sqrt{2}  + \varepsilon < 1.586+\varepsilon$.

        If furthermore each block of the given graph is degree three bounded, there is a
        polynomial time approximation algorithm for \HPP with performance
        guarantee $1.5 + \varepsilon$, for any $\varepsilon>0$.
    \end{theorem}
    The generalization to the traveling salesman problem is presented in
    Section~\ref{sec:tspp}.

\section{Preliminaries}\label{sec:prelim}
    \vspace{-0.1cm}
    \paragraph{Held-Karp Relaxation.} The linear program known as the
    Held-Karp (or subtour elimination) relaxation is a well studied lower
    bound on the value of an optimal tour. It has a variable $x_{\{u,v\}}$
    for each pair of vertices with the intuitive meaning that
    $x_{\{u,v\}}$ should take value $1$ if the edge $\{u,v\}$ is used in
    the tour and $0$ otherwise. Letting $G=(V,E)$ be the complete graph on
    the set of vertices and $c_{\{u,v\}}$ be the distance between vertices
    $u$ and $v$, the Held-Karp relaxation can then be formulated as the
    linear program where we wish to  minimize $\sum_{e \in E} c_{e} x_{e}$ subject to
    $$  x(\delta(v)) =
      2 \mbox{    for }v\in V  \mbox{,}  \qquad x(\delta(S)) \geq 2 \mbox{   for } \emptyset \neq S \subset V, \qquad \mbox{ and $x\geq 0$},
    $$
    where $\delta(S)$ denotes the set of edges crossing the cut $(S, \bar S)$ and
    $x(F) = \sum_{e\in F} x_e$ for any $F\subseteq E$.

    Goemans \& Bertsimas~\cite{GB90} proved that for metric distances the
    above linear program has the same optimal value as the linear program
    obtained by dropping the equality constraints. Moreover, when
    considering a \TSP instance $G=(V,E)$ we only need to consider the
    variables $(x_{e})_{e\in E}$. Indeed, any solution $x$ to the
    Held-Karp relaxation without equality constraints such that $x_{\{u,v\}}>0$
    for a pair of vertices $\{u,v\} \not \in E$ can be transformed into a
    solution $x'$ with no worse cost and $x'_{\{u,v\}} = 0$ by setting
    $x'_{e} = x_{e} + x_{\{u,v\}}$ for each edge on the shortest path
    between $u$ and $v$, and $x_e' = x_e$ for the 
    other edges.  The Held-Karp relaxation for \TSP on a graph $G=(V,E)$
    can thus be formulated as follows:
    $$
     \min \sum_{e \in E} x_{e} \qquad \mbox{ subject to}
    \qquad x(\delta(S)) \geq 2 \mbox{ for } \emptyset \neq S \subset V,\qquad \mbox{and } x\geq 0.
    $$
    We shall refer to this linear program as \LP{G} and denote the value
    of an optimal solution by \OLP{G}.  Its integrality gap was previously known to
    be at most $3/2- \epsilon$ and at least $4/3$ for graphic instances.
    The lower bound
    is obtained by a claw-free graphic instance of degree at most $3$ that
    consists of three paths of equal length with endpoints $(s_1, t_1),
    (s_2, t_2),$ and $(s_3,t_3)$ that are connected so as $\{s_1, s_2,
    s_3\}$ and $\{t_1, t_2, t_3\}$ form two triangles (see Figure~\ref{fig:intgap}).

    We end our discussion of \LP{G} with a useful observation. When
    considering \TSP{}, it is intuitively clear that we can restrict
    ourselves to \emph{$2$-vertex-connected} graphs, i.\,e., graphs that
    stay connected after deleting a single vertex. Indeed, if we consider
    a graph with a vertex $v$ whose removal results in components $C_1,
    \ldots, C_\ell$ with $\ell >1$ then we can recursively solve the
    \TSP{} problem on the $\ell$ subgraphs $G_1, G_2, \dots, G_\ell$
    induced by $C_1 \cup \{v\}, C_2 \cup \{v\}, \dots, C_\ell \cup
    \{v\}$. The union of these solutions will then provide a solution to
    the original graph that preserves the approximation guarantee with
    respect to the linear programming relaxation since one can see that
    $\OLP{G} \geq \sum_{i=1}^\ell \OLP{G_i}$. We summarize this
    observation in the following lemma (see Appendix~\ref{app:2connTSP}
    for a fullproof).
    \begin{lemma}
    \label{lemma:2connTSP}
      Let $G$ be a connected graph. If there is an $r$-approximation
      algorithm for \TSP{} on each $2$-vertex-connected subgraph $H$ of $G$ (with
      respect to \OLP{H}) then there is an $r$-approximation algorithm for \TSP on 
      $G$ (with respect to \OLP{G}).
    \end{lemma}

    \vspace{-0.4cm}
\paragraph{Matchings of Cubic $2$-Edge-Connected Graphs.}
    Edmonds~\cite{Edmonds1965b} showed that the following set of
    equalities and inequalities on the variables $(x_e)_{e\in E}$
    determines the perfect matching polytope (i.\,e.,  all extreme points of the polytope are integral and correspond to perfect matchings)
    of a given graph $G=(V, E)$:
    $$
      x(\delta(v))  = 1 \mbox{ for } v\in V,\qquad
      x(\delta(S))  \geq 1 \mbox{ for } S \subseteq V\mbox{ with $|S|$ odd,}\qquad \mbox{ and } 
      x  \geq 0.
    $$
    The linear description is useful for understanding the structure of
    the perfect matchings. For example, Naddef and Pulleyblank~\cite{NP81}
    proved that $x_e = 1/3$ defines a feasible solution when $G$ is
    \emph{cubic} and \emph{$2$-edge connected}, i.\,e., every vertex has
    degree $3$ and the graph stays connected after the removal of an
    edge. They used that result to deduce that such graphs always have a
    perfect matching of weight at least $1/3$ of the total weight of the
    edges.

    Standard algorithmic versions of Carath\'{e}odory's theorem (see
    e.\,g. Theorem~$6.5.11$ in~\cite{GLS1988}) say that, in polynomial time, we can
    decompose a feasible solution to the perfect matching polytope into a
    convex combination of polynomially many perfect matchings (see
    also~\cite{Barahona04} for a combinatorial approach for the matching
    polytope).
    Combining these results leads to the following lemma (see~\cite{BSSS11,GLS05,MMP90} for
    closely related variants that also have  been useful for the
    \TSP problem).
    \begin{lemma}
    \label{lem:matching}
      Given a cubic $2$-edge-connected graph $G$, we can in polynomial time
      find a distribution over polynomially many perfect matchings so that
      with probability $1/3$ an edge is in a  perfect matching picked from this distribution.
    \end{lemma}
    Note that all $2$-vertex-connected graphs except the trivial graph on
    $2$ vertices are $2$-edge connected. We can therefore apply the above
    lemma to cubic $2$-vertex-connected graphs.

    \vspace{-0.2cm}
\section{Approximation  Framework}
\label{sec:tspframework}
Lemma~\ref{lemma:2connTSP} says that the technical difficulty in
approximating the \TSP problem lies in approximating those instances
that are $2$-vertex connected.
As alluded to in the introduction, we shall generalize previous
results~\cite{FJJ89,MMP90} that relate the cost of an optimal tour to
the size of a minimum $2$-vertex-connected subgraph. The main
difference is the use of matchings. Traditionally, matchings have been
used to add edges to make a given graph Eulerian whereas our framework
offers a structured way to specify a set of edges that safely may be
removed leading to a lower cost. To identify the set of edges that may
be removed we use the following definition.

\vspace{-0.1cm}
\begin{definition}[Removable pairing of edges] \label{def:pairing}
  Given a $2$-vertex-connected graph $G$ we call a tuple $(R,P)$ consisting of a  subset $R$ of removable edges and  a subset   $P\subseteq R\times R$ of pairs of edges a  \emph{\MS{}} if
\vspace{-0.2cm}
  \begin{itemize}\itemsep-1mm
    \item an edge is in at most one pair;
    \item the edges in a pair are incident to a common vertex of degree at least $3$;
    \item any graph obtained by deleting removable edges so that at most one edge in each pair is deleted stays connected.
  \end{itemize}
\end{definition}

The following theorem generalizes the corresponding result of~\cite{MMP90} (their
result follows from the the special case of an empty removable pairing).
\begin{theorem}
\label{thm:main}
Given a $2$-vertex-connected graph $G=(V,E)$ with a \MS{} $(R,P)$,
there is a polynomial time algorithm that returns a
 spanning Eulerian
multigraph in $G$ with at most
  $\frac{4}{3}\cdot |E| - \frac{2}{3} \cdot |R|$ edges.
\end{theorem}
The proof of the  theorem is presented after the following lemma on which it is based.
\begin{lemma}
\label{lemma:sample}
Given a $2$-vertex-connected graph $G=(V,E)$ with a \MS{} $(R,P)$, we
can in polynomial time find a distribution over polynomially many
subsets of edges such that a random subset $M$ from this distribution
satisfies:
\vspace{-0.2cm}
\begin{itemize}\itemsep-1mm
\item[(a)] each edge is in $M$ with probability $1/3$;
\item[(b)] at most one edge in each pair is in $M$; and
\item[(c)] each vertex has an even degree in the multigraph with edge set $E
  \cup M$.
\end{itemize}
\end{lemma}
\begin{proof}
  We shall use Lemma~\ref{lem:matching} and will therefore need a
  cubic $2$-edge-connected graph. In the spirit of~~\cite{FJJ89}, we replace
  all vertices of $G$ that are not of degree three by gadgets to
  obtain a cubic graph $G'=(V',E')$ as follows (see also
  Figure~\ref{fig:degreplace}):
\begin{itemize}
\item A vertex $v$ of degree 2 with neighbors $u$ and $w$ is replaced
  by a cycle consisting of four vertices $v_N$, $v_W$, $v_S$, $v_E$
  with the chord $\{v_W, v_E\}$. The gadget is then connected to the
  neighbours of $v$ by the the edges $\{u,v_N\}$ and $\{v_S,w\}$.
    
\item A vertex $v$ with $d(v) > 3$ is replaced by a tree $T_v$ that
  has $\lfloor d(v)/2 \rfloor$ leaves, a binary root if $d(v)$ is odd,
  and otherwise only degree $3$ internal vertices. Each leaf is
  connected to two neighbours of $v$ such that the edges incident to $v$
  that form a pair in $P$ are incident to the same leaf.  If $d(v)$ is
  odd, one of the neighbors is left and connected to the binary root.
\end{itemize}

\begin{figure}[bt]
\begin{center}
\includegraphics[width=14cm]{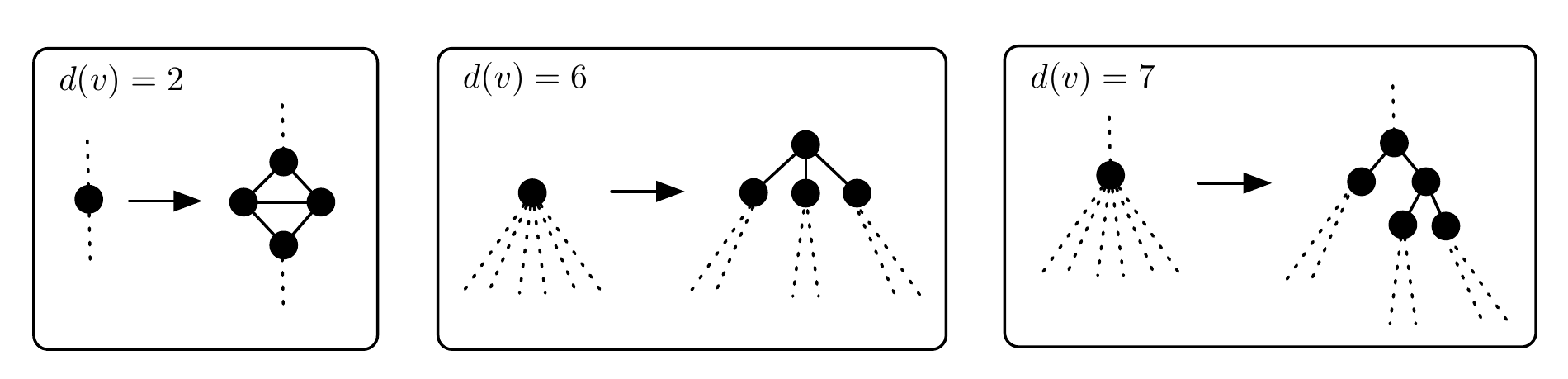}
\end{center}
\caption{Examples of the used gadgets to obtain a cubic graph.}
\label{fig:degreplace}
\end{figure}
The above gadgets guarantee that the graph $G'$ is cubic and it is
$2$-vertex connected since $G$ was assumed to be $2$-vertex connected.
We can therefore apply Lemma~\ref{lem:matching} in order to obtain
a random perfect matching $M'$. Each edge of $G'$ is in $M'$ with a
probability of exactly 1/3. Let $M$ be the set of edges obtained by
restricting $M'$ to the edges of $G$ in the obvious way. Now $M$
contains each edge of $G$ with probability
1/3.
We complete the proof by showing that $M$ also satisfies properties
$(b)$ and $(c)$.  As each pair of edges in $P$ is incident to a vertex
of degree at least $3$, we have, by the construction of the gadgets,
that they are incident to a common vertex in $G'$ and hence at most
one edge of each pair is in $M$. 
Finally, property $(c)$ follows from that $E' \cup M'$ is clearly a spanning Eulerian
multigraph of $G'$ and compressing a set of even-degree vertices results in one vertex of even degree.
\end{proof}

Equipped with the above lemma we are now ready to prove the main result
of this section.
\begin{proofof}{Theorem~\ref{thm:main}}
  Pick a random subset $M \subseteq E$ of edges that satisfies the
  properties of Lemma~\ref{lemma:sample}. Let $M_R$ be the set of those edges of
  $M$ that are removable and let $\bar M_R$ be the set of the remaining edges of
  $M$.

Consider the multigraph $H$ on vertex set $V$ and edge set $E
\setminus M_R \cup \bar M_R$. Observe that both adding an edge and
removing an edge swaps the parity of the degree of an incident
vertex. We have thus from property $(c)$ of Lemma~\ref{lemma:sample}
that the degree of each vertex in $H$ is even. Moreover, as $(R,P)$ is a
removable pairing, property $(b)$ of Lemma~\ref{lemma:sample}
gives that $H$ is connected. Alltogether we have that $H$ is an
Eulerian graph, i.\,e., a \TSP{} solution. We continue to calculate its
expected number of edges, which is
\begin{equation}
\label{eq:nredges}
\mathbb{E}[|E| +   |\bar M_R| -
|M_R|].
\end{equation} 
Using that each edge is in $M$ with probability $1/3$,  we have, by linearity of
expectation, that~\eqref{eq:nredges} equals
$$
|E| + \frac{1}{3} (|E| - |R|) - \frac{1}{3} |R| = \frac{4}{3} \cdot |E| - \frac{2}{3} \cdot |R|.
$$ 

To conclude the proof, we note that the selection of $M$ can be
derandomized since there are, by Lemma~\ref{lemma:sample},
polynomially many edge subsets to choose from; taking the one that
minimizes the number of edges of $H$  is sufficient.
\end{proofof}

\section{Finding a Removable Pairing by Minimum Cost Circulation}\label{sec:circulation}
In order to use our framework, one of the main challenges is to find a
\MS{} that is sufficiently large. In the following, we show how to obtain a
useful \MS{} based on circulations.

Consider a $2$-vertex connected graph $G$ and let $T$ be a spanning
tree of $G$ obtained by depth-first search (starting from some
arbitrary root $r$). Then each edge in $G$ connects a vertex to either
one of its predecessors or one of its successors. We call the edges in
$T$ \emph{tree-edges} and those in $G$ but not in $T$
\emph{back-edges}.

We shall now define a circulation network $C(G,T)$.
We start
by introducing an orientation of $G$: all tree-edges
become tree-arcs directed from the root to the leaves and all 
back-edges become back-arcs directed towards the root. To distinguish the circulation
network and the original graphs, we use the names $\oa{G}$ and $\oa{T}$ for the
network versions of $G$ and $T$.  In order to ensure connectivity properties of
subnetworks obtained from feasible circulations, we replace some of the vertices
by gadgets. 

For each vertex $v$ except the root that has $\ell$ children $w_1,
w_2, \ldots, w_\ell$ in the tree, we introduce $\ell$ new vertices $v_1$,
$v_2$, $\ldots$, $v_\ell$ and replace the tree-arc $(v, w_j)$ by the tree-arcs $(v,v_j)$ and $(v_j, w_j)$ for $j= 1,2, \ldots, \ell$. 
Then we redirect all incoming back-arcs of $v$ from the subtree rooted
by $w_j$ to $v_j$. For an illustration of the gadget see
Figure~\ref{fig:circreplace} and for an example of a complete network
see Figure~\ref{fig:circreplace_appendix}.  This way,
all back-arcs start in old vertices and lead to new vertices or the
root.  In the following, we call the new vertices and the root
\emph{in-vertices} and the remaining old ones \emph{out-vertices}. We
also let $\mathcal{I}$ be the set of all in-vertices.

\begin{figure}[bt]
\begin{center}
  \includegraphics[width=10cm]{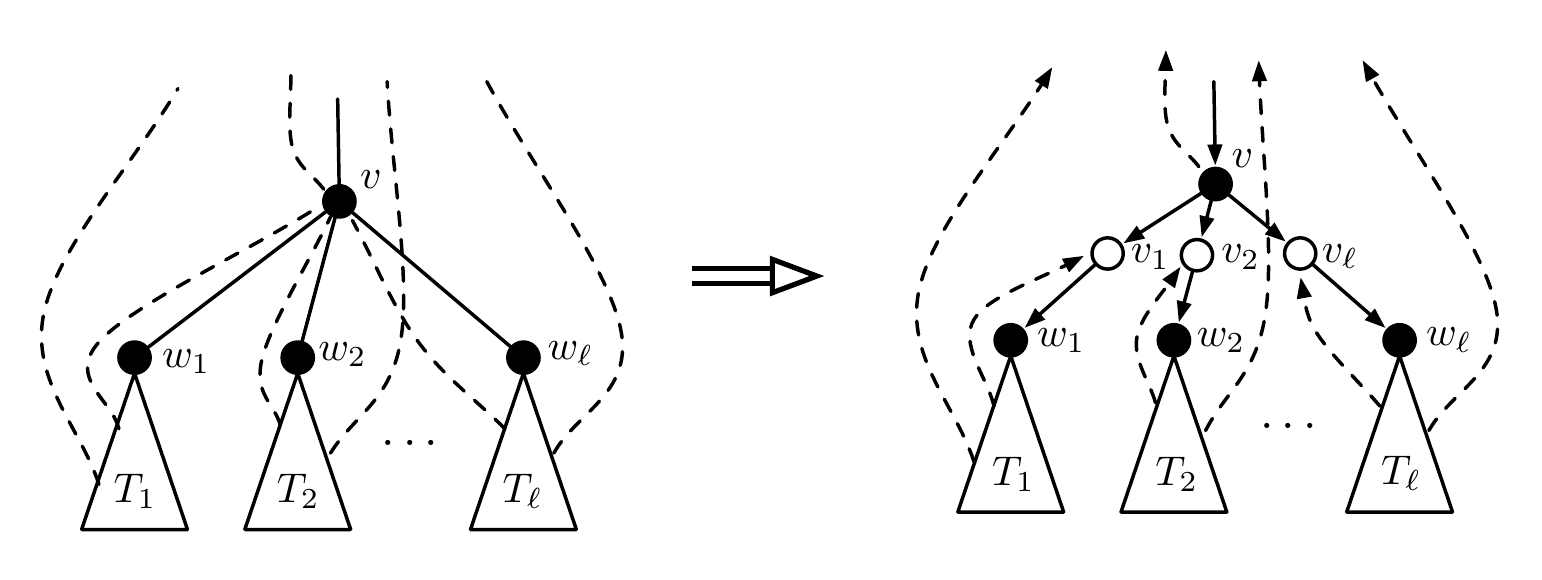}
\end{center}
\caption{The gadget that, for each child of $v$, introduces a new vertex (depicted in
  white) and redirects back-arcs.}
\label{fig:circreplace}
\end{figure}

We now specify
a lower bound (demand) and an upper bound (capacity) on the circulation.
For each arc $a$ in $\oa{T}$, we set the demand of $a$ to 1 and for
all other arcs to
0.
The capacity is $\infty$ for any arc.  Finally, the cost of a circulation $f$
in $C(G,T)$ is the piecewise linear function $\sum_{v\in\mathcal{I}}
\max[f(B(v))-1, 0]$, where $B(v)$ is the set of incoming back-arcs of
$v$. One can think of the cost as the total circulation on the
back-arcs except that each in-vertex accepts a circulation of $1$ for
free.  Note that algorithmically there is no considerable difference
whether we use our cost function or define a linear cost function on
the arcs: for any in-vertex $v$ we can redirect all back-arcs of $v$
to a new vertex $v'$ and introduce two arcs $(v',v)$, one of cost $0$
and capacity $1$ and the other  of cost $1$ and capacity $\infty$. All
remaining arcs then have a cost of $0$.

The following lemma shows how to use a circulation in $C(G,T)$ to approximate
\TSP.
\begin{lemma}
\label{lemma:costcirc}
Given a $2$-vertex connected graph $G$ and a depth first search tree
$T$ of $G$ let $C^*$ be the minimum cost circulation to $C(G,T)$ of
cost $c(C^*)$. Then there is a spanning Eulerian multigraph $G'$ in
$G$ with at most $\frac{4}{3} n + \frac{2}{3} c(C^*) - 2/3$ edges.
\end{lemma}
\begin{proof}
  \newcommand{\flowT}{\ensuremath{\oa{T}}}
  \newcommand{\graphC}{\ensuremath{{C^*(G,T)}}} We first note that,
  for any arc of $C(G,T)$, the demand and the capacity is integral.
  Therefore, applying Hoffman's circulation theorem (see \cite{Sch03},
  Corollary 12.2a), we can assume the circulation $C^*$ to be
  integral.  Let $\graphC$ be the support of $C^*$ in $C(G,T)$,
  i.\,e., the induced subgraph of the arcs with non-zero circulation
  in $C^*$, and let $G'$ be the subgraph of $G$ obtained from
  $\graphC$ by compressing the gadges of the circulation network in
  the obvious way. 

  To prove the lemma, we shall first prove that graph $G'$ is
  $2$-vertex connected and then define a removable pairing $(R,P)$ on
  $G'$ in order to apply Theorem~\ref{thm:main}. That $G'$ is
  $2$-vertex connected follows from flow conservation, that each arc
  $a$ in $\oa{T}$ has demand $1$, and the design of the
  gadgets. Indeed, if $G'$ would have a cut vertex $v$ with children
  $w_1, w_2, \ldots, w_\ell$ in $T$ then one of the subtrees, say the
  one rooted by $w_j$, has no back-edges to the ancestors of $v$ which
  in turn, by flow conservation, would contradict that the tree-arc
  $(v, v_j)$ in $\oa{T}$ carries a flow of at least $1$. (Recall that
  the edge $\{v, w_j\}$ in $T$ is replaced by tree-arcs $(v, v_j)$ and
  $(v_j, w_j)$ in $\oa{T}$.)

  We now determine a \MS $(R,P)$ on $G'$. For ease of argumentation we
  shall first slightly abuse notation and define a \MS $(R_C, P_C)$ on
  $\graphC$.
  The set $P_C$ consists of all $(e,e')$ such that $e=(u,v)$ is a
  back-arc of cost zero in $\graphC$, $v$ has at least two incoming
  arcs, and $e'=(v,w)$ is a tree-arc.  Note that each such $v$ is an
  in-vertex, the number of incoming back-arcs of cost zero is at most
  one, $e'$ is the unique outgoing tree-arc of $v$, and the only
  possible vertex $v$ with only one incoming back-arc and no other
  incoming arc is the root.
  The set $R_C$ contains all edges from $P_C$ and additionally all
  remaining back-arcs of $\graphC$.  In other words, each edge of
  $\graphC$ that is neither in $\oa{T}$ nor in $P$ is a back-arc with
  integer non-zero cost in the circulation or a back-arc to the
  root. Hence, $|R_C| - 2|P_C| = c(C^*)$ if the root has more than one
  incoming back-arc and $|R_C|- 2|P_C| = c(C^*) +1$ otherwise.

  The \MS $(R,P)$ on $G'$ is now obtained from $(R_C, P_C)$, by
  mereley compressing the gadgets used to form $C(G,T)$ and by
  dropping the orientations of the arcs. As all edges in $R_C$ are
  either back-arcs or they are tree-arcs starting from an in-vertex,
  no arc in $R_C$ is removed by the compression and thus $|R|=|R_C|$
  and $|P| = |P_C|$. Moreover, $G'$ has $(n-1) + |R| - |P|$
  edges and, assuming $(R,P)$ is a valid \MS, Theorem~\ref{thm:main} 
  yields that $G'$ (and thus $G)$ has a spanning Eulerian multigraph
  with at most $ \frac{4}{3} ((n-1) + |R|-|P|) - \frac{2}{3} |R| =
  \frac{4}{3} n + \frac{2}{3}( |R| - 2|P|) - \frac{4}{3} \leq
  \frac{4}{3} n + \frac{2}{3} c(C^*) - \frac{2}{3} $ edges. The
  last inequality followed  from that $|R| - 2|P|$  is at most
  $c(C^*) +1$.
 
  Therefore, we can conclude the proof by showing that $(R,P)$ is a
  valid \MS.  It is easy to verify that $(R,P)$ satisfies the first
  two conditions of Definition~\ref{def:pairing}, that is, each edge
  is contained in at most one pair and the edges in each pair are
  incident to one common vertex of degree at least three. The third
  condition follows from that, for any vertex $v$ of $G'$, the
  vertices in the subtree $T_v$ of $T$ rooted by $v$ form a connected
  subgraph of $G'$ even after removing edges according to $(R,P)$. To
  see this we do a simple induction on the depth of $v$. In the base
  case, $v$ is a leaf and the statement is clearly true. For the
  inductive step, consider a vertex $v$ with $\ell$ children $w_1,
  w_2, \dots, w_\ell$ in $T$. By the inductive hypothesis, the
  vertices in $T_{w_j}$ for $j=1,2, \dots, \ell$ stay connected after
  the removal of edges according to $(R,P)$. To complete the inductive
  step it is thus sufficient to verify that $v$ is connected to each
  $T_{w_j}$ after the removal of edges. If $\{v,w_j\}$ is not in $R$
  this clearly holds. Otherwise if $e_j = \{v, w_j\} \in R$ then by
  the definition of $(R, P)$ there is an edge $e$ such that $(e, e_j)
  \in P$ and $e$ is incident to $v$ and a vertex in $T_{w_j}$. Since
  at most one edge in each pair is removed we have that $v$ also stays
  connected to $T_{w_j}$ in this case, which completes the inductive
  step. We have thus proved that $(R,P)$ satisfies the properties of a
  \MS which completes the proof of the statement.

\end{proof}

\section{Improved Approximation Algorithms}
\label{sec:algorithms}
We first show how to apply our framework to restricted graph classes
for which we obtain a tight bound on the integrality gap of the
Held-Karp relaxation. We then show how to use our framework to obtain
an improved approximation algorithm for general graphs.

\subsection{Bounded Degree and Claw-Free Graphs}
We consider the class of graphs that have a degree bounded by three.
\begin{lemma}\label{lem:boundeddeg}
  Given a $2$-vertex-connected graph $G$ with $n$ vertices, there is a
  polynomial time algorithm that computes a spanning Eulerian
  multigraph $H$ in $G$ with at most $4n/3 - 2/3$ edges.
\end{lemma}
\begin{proof}
  If $G$ has one or two vertices, we obtain an Eulerian multigraph of
  zero or two edges. Otherwise, we compute a depth-first search tree
  $T$ in $G$ and determine the circulation network $C(G,T)$. We now
  show that this network has a feasible circulation $f$ of cost at
  most one. Let us assign a circulation of one to each back-arc $e$ in
  $C(G,T)$ and push it through the path in $\oa{T}$ that is incident
  to both the start and end vertex of $e$. By the construction of
  $C(G,T)$ and from the assumption that $G$ is $2$-vertex connected,
  each tree-arc is in a directed cycle that contains exactly one
  back-arc. Therefore, all demand constraints are satisfied. Due to
  the degree-bounds, no vertex but the root has more than one incoming
  back-arc. The cost $\sum_{v\in\mathcal{I}} \max[f(B(v))-1, 0]$ of
  the circulation is therefore at most one and zero if the root has
  only one back-arc.  If the circulation cost is zero, by
  Lemma~\ref{lemma:costcirc} we obtain a spanning Eulerian multigraph
  $H$ in $G$ with at most $4n/3-2/3$ edges. For those circulations
  where the cost is one, the proof of Lemma~\ref{lemma:costcirc}
  allows to save an additional constant of $2/3$ (since then the root
  has more than one incoming back-arc) and we obtain the same bound on
  the number of edges.

\end{proof}
Note that it is sufficient to find a 2-vertex-connected
degree three bounded spanning subgraph (a 3-trestle) and thus, using a
result from \cite{KKN01}, we can apply Lemma~\ref{lem:boundeddeg} also
to claw-free graphs.  Applying Lemma~\ref{lemma:2connTSP}, we obtain
an upper bound of 4/3 on the integrality gap for the Held-Karp
relaxation for the considered class of graphs. In addition, along the
lines of the proof of Lemma~\ref{lemma:2connTSP}, one can see that the
above arguments imply that any connected graph $G$ decomposed into $k$
blocks, i.\,e., maximal $2$-connected subgraphs, such that each block is
either degree three bounded or claw-free, has a spanning Eulerian
multigraph with at most $4n/3 + 2k/3- 4/3$ edges.

\subsection{General Graphs}
We now apply our framework to graphs without degree constraints.  We
start with an algorithm that achieves an approximation ratio better
than $3/2$ for graphs for which the linear programming relaxation has
a value close to $n$.  Let $G= (V,E)$ be an $n$-vertex graph.  The
support $E' = \{e: x_e^* > 0\}$ of an extreme point $x^*$ of \LP{G} is
known to contain at most $2n-1$ edges~(see Theorem~$4.9$
in~\cite{CFN85}). Moreover, if we let $x^*$ be an optimal solution,
then any $r$-approximate solution to graph $G'=(V,E')$ with respect to
\OLP{G'} is an $r$-approximate solution to $G$ with respect to
\OLP{G}, because $E' \subseteq E$ and $\OLP{G'} = \OLP{G}$. We can
thus restrict ourselves to $n$-vertex graphs with at most $2n-1$ edges
and, by Lemma~\ref{lemma:2connTSP}, we can further assume the graph to
be $2$-vertex connected.
\begin{algorithm}[h]
\begin{algorithmic}[1]
\REQUIRE A 2-vertex-connected graph $G$ with $n$ vertices and at most $2n-1$ edges.
\STATE Obtain an optimal solution $x^*$ to \LP{G}.
\STATE Obtain a depth-first-search tree $T$ of $G$ by starting at some
  root and in each iteration pick, among the possible edges, the edge $e$
  with maximum $x_e^*$.
\STATE Solve the min cost circulation problem $C(G,T)$ to obtain a circulation $C^*$ with cost $c(C^*)$.

\STATE Apply Lemma~\ref{lemma:costcirc} to find a spanning Eulerian
  multigraph with less than $\frac{4}{3} n + \frac{2}{3} c(C^*)$ edges.
\end{algorithmic}
\caption{}
\label{alg:allgraphs}
\end{algorithm}

To analyze the approximation ratio achieved by
Algorithm~\ref{alg:allgraphs}, we bound the cost of the circulation.
\begin{lemma}
\label{lemma:circcost}
  We have $c(C^*) \leq 6(1-\sqrt{2})n + (4\sqrt{2} -3)\OLP{G}$.
\end{lemma}
\begin{proof}
  For notational convenience, when considering an arc $a$ in the flow
  network we shall slightly abuse notation and use $x^*_a$ to denote
  the value of the corresponding edge in $G$ according to the optimal
  LP-solution $x^*$.
  We prove the statement by defining a fractional circulation $f$ of
  cost at most $6(1-\sqrt{2})n + (4\sqrt{2} -3)\OLP{G}$. The
  circulation $f$ will in turn be the sum of two circulations $f'$ and
  $f''$. We obtain the circulation $f'$ as follows: for each back-arc
  $a$ we push a flow of size $\min[x^*_a,1]$ along the cycle formed by
  $a$ and the tree-arcs in $\oa{T}$.  We shall now define the circulation
  $f''$ so as to guarantee that $f$ forms a feasible circulation,
  i.\,e., one that satisfies the demands $f_a \geq 1$ for each $a\in
  \oa{T}$. As out- and in-vertices are alternating in $\oa{T}$ and
  in-vertices have only one child in $\oa{T}$ and no outgoing
  back-edges, a sufficient condition for $f$ to be feasible can be
  seen to be $f_a \geq 1$ for each $a\in \oa{T}$ that is from an
  out-vertex to an in-vertex. To ensure this, we now define $f''$ as
  follows. For each vertex $v$ of $G$ that is replaced by a gadget
  consisting of an out-vertex $v$ and a set $\mathcal{I}_v$ of
  in-vertices, we push for each $w\in \mathcal{I}_v$ a flow of size
  $\max[1-f'_{(v,w)}, 0]$ along a cycle that includes the arc $(v,w)$
  (and one back-arc).
  Note that such a cycle is guaranteed to exist since $G$ was assumed to be
  $2$-vertex connected.
  From the definition of $f''$, we have thus
  that $f=f'+f''$ defines a feasible circulation.

  We proceed by analyzing the cost of $f$, i.\,e., $\sum_{v\in
    \mathcal{I}} \max[f(B(v)) -1, 0]$, where $\mathcal{I}$ is the set
  of all in-vertices and $B(v)$ is the set of incoming back-arcs of
  $v\in \mathcal{I}$. Note that the cost is upper bounded by
  $\sum_{v\in \mathcal{I}} \max[f'(B(v)) -1, 0] + \sum_{v\in
    \mathcal{I}} f''(B(v))$ and we can thus analyze these two terms
  separately. We start by bounding the second summation and then
  continue with the first one. If $\OLP{G} = n$ then one can see that $f'' = 0$. Moreover,
  \begin{claim}
    \label{claim:firstcost}
    We have $\sum_{v\in \mathcal{I}} f''(B(v)) \leq \OLP{G}-n$.
  \end{claim}
  \begin{proofclaim}
    When considering a vertex $v$ as done above in the definition of
    $f''$, the flow pushed on back-arcs is $\sum_{w\in \mathcal{I}_v}
    \max[1-f'_{(v,w)}, 0]$ which equals $\sum_{w\in \mathcal{I}'_v}
    (1-f'_{(v,w)})$, where $\mathcal{I}'_v = \{w \in \mathcal{I}_v: f'_{(v,w)} < 1\}$. Letting $T_w$ be the set of vertices of $G$ in the
    subtree of the undirected tree $T$ rooted by the child of  $w\in \mathcal{I}'_v$,
    we have, by the definition of $f'$,
$$
f'_{(v,w)} = \sum_{a\in \delta(T_w)\setminus \delta(v)} \min[x^*_a, 1]
= x^*(\delta(T_w) \setminus \delta(v)).
$$
The second equality follows from that if $x^*_a > 1$ for some $a\in
\delta(T_w) \setminus \delta(v)$ then $f'_{(v,w)} \geq 1$ and hence $w\not \in \mathcal{I}'_v$. We have thus 
$
\sum_{w\in \mathcal{I}'_v} (1- f'_{(v,w)}) = |\mathcal{I}'_v| -
\sum_{w\in \mathcal{I}'_v} x^*(\delta(T_w) \setminus \delta(v)).
$
As we are considering a depth-first-search tree (see
Figure~\ref{fig:circcostOLA}),
\begin{align}
\label{eq:equalsums}
2 \sum_{w\in \mathcal{I}'_v} x^*(\delta(T_w) \setminus \delta(v)) &=  
\sum_{w\in \mathcal{I}'_v} x^*(\delta(T_w)) +
x^*\left(\delta\left(\bigcup_{w\in \mathcal{I}'_v} T_w \cup
    \{v\}\right)\right) - x^*(\delta(v)).
\end{align}
\begin{figure}[bt]
\begin{center}
\includegraphics[width=7cm]{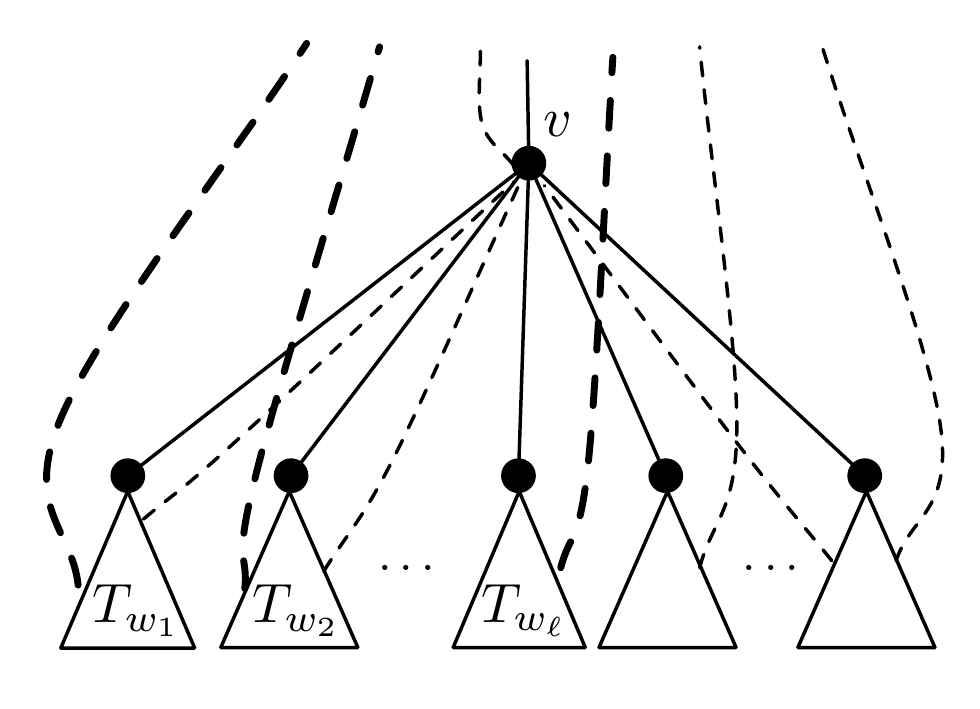}
\end{center}
\caption{An illustration of Equality~\eqref{eq:equalsums} with $\mathcal{I}'_v = \{w_1, w_2, \dots, w_\ell\}$: both the
  left-hand-side and the right-hand-side of the equality express two
  times the value of the fat edges.}
\label{fig:circcostOLA}
\end{figure}
Since by the feasibility of $x^*$ each of the sets corresponds to a cut of fractional value at least $2$ we use $2\cdot (|\mathcal{I}'_v| + 1) - x^*(\delta(v))$ as a lower bound on~\eqref{eq:equalsums}.

\noindent Summarizing the above calculations yields
$$
\sum_{w\in \mathcal{I}'_v} \left(1- f'_{(v,w)}\right) = |\mathcal{I}'_v| -
\sum_{w\in \mathcal{I}'_v} x^*(\delta(T_w) \setminus \delta(v)) \leq \frac{x^*(\delta(v))}{2} - 1.
$$

Repeating this argument for each $v$ we have $\sum_{v\in \mathcal{I}}
f''(B(v)) = \sum_{v \in V} \sum_{w \in \mathcal{I}'_v} \left(1-f'_{(v,w)}\right) \le
\sum_{v\in V} \left(\frac{x^*(\delta(v))}{2} -1\right)$,
which equals $\OLP{G}-n$ since $\OLP{G} = \frac{1}{2}\sum_{v\in V}
x^*(\delta(v))$.
\end{proofclaim}
We proceed by bounding $\sum_{v\in \mathcal{I}} \max[f'(B(v)) -1, 0]$ from above.  
\begin{claim}
\label{claim:secondcost}
We have $\sum_{v\in \mathcal{I}} \max[f'(B(v)) -1, 0] \leq (7-6\sqrt{2})n + 4(\sqrt{2}-1)\OLP{G}$
\end{claim}
\begin{proofclaim}
To analyze this expression we shall use two facts. First
$G$ has at most $2n-1$ edges, and therefore the number of back-arcs
is at most $2n-1 - (n-1) = n$.  Second, as the depth-first-search
chooses (among the available edges) the edge $a$ with maximum $x^*_a$ in
each iteration, we have that $x_{a}^* \leq x_{t_v}^*$ for each $a \in
B(v)$ where $t_v$ is the outgoing tree-arc of $v\in
\mathcal{I}$. Moreover, as $f'_a = \min[x^*_{a}, 1 ]$ for each
back-arc, the number of back-arcs in $B(v)$ is at least $\left\lceil
  \frac{f'(B(v))}{\min[x^*(t_v),1]}\right\rceil$.
Combining these two facts gives us that
 \begin{equation}
\label{eq:cap2}
\sum_{v\in \mathcal{I}}
\left\lceil \frac{f'(B(v))}{\min[x^*(t_v),1]}\right\rceil \leq n. 
\end{equation}
For $v\in \mathcal{I}$, we partition $f'(B(v))$ into $\ell_v =
\min[2-x^*(t_v), f'(B(v))]$ and $u_v = f'(B(v)) -
\ell_v$. Furthermore, let $u^* = \sum_{v\in \mathcal{I}} u_v$. With
this notation we can upper bound $\sum_{v\in \mathcal{I}}
\max[f'(B(v)) -1, 0]$ by
\begin{equation}
\label{eq:fcost}
\sum_{v\in \mathcal{I}} \max[\ell_v - 1, 0] + u^*
\end{equation}
and relax Inequality~\eqref{eq:cap2} to
\begin{equation}
\label{eq:cap}
\sum_{v\in \mathcal{I}} \frac{\ell_v}{x^*(t_v)} \leq  n- u^*.
\end{equation}

The cost~\eqref{eq:fcost} (where we ignore $u^*$) subject
to~\eqref{eq:cap} can now be interpreted as a knapsack problem of
capacity $n - u^*$ that is packed with an item of profit $ \max[\ell_v
-1, 0]$ and size $\ell_v/x^*(t_v)$ for each $v\in
\mathcal{I}$. Consequently, we can upper bound~\eqref{eq:fcost} by
considering the fractional knapsack problem with capacity $n-u^*$ and
infinitely many items of a maximized profit to size ratio. Associating
a variable $L$ with $\ell_v$ and $T$ with $x^*(t_v)$ this ratio is $
\max_{0\leq T \leq 1,0\leq L \leq 2 - T} \frac{L-1}{L}\cdot T.$ For
any $T$ the ratio is maximized by letting $L=2-T$ and we can thus
restrict our attention to items with profit to size ratio
$\max_{0\leq T \leq 1} \frac{1-T}{2-T} \cdot T$. A simple analysis (see Appendix~\ref{sec:maxsizratio})
shows that the maximum is achieved when $T= 2-
\sqrt{2}$. Therefore, the profit~\eqref{eq:fcost} is upper bounded by
$$
 \frac{\sqrt{2} -1}{\sqrt{2}} \cdot(2- \sqrt{2})\cdot (n-u^*) +u^* = (\sqrt{2}-1)^2 \cdot (n-u^*) + u^*.
$$
As the fractional degree of a vertex $v$ that is replaced by a gadget
with a set $\mathcal{I}_v$ of in-vertices is at least $2+\sum_{w\in
  \mathcal{I}_v} u_w$, we have $u^* \leq 2(\OLP{G}-n)$. Hence,
$$
\eqref{eq:fcost} \leq  (\sqrt{2}-1)^2 \cdot (n-2(\OLP{G}-n)) + 2(\OLP{G}-n),
$$
which equals $(7-6\sqrt{2})n + 4(\sqrt{2}-1)\OLP{G}$.
\end{proofclaim}
Finally, by summing up the bounds given by Claim~\ref{claim:firstcost}
and Claim~\ref{claim:secondcost} we bound the cost of $f$ and hence $c(C^*)$ from above by
$
 \OLP{G}-n + n(7-6\sqrt{2}) + 4(\sqrt{2}-1)\OLP{G},
$
which equals $6(1-\sqrt{2})n + (4\sqrt{2} -3)\OLP{G}$.
\end{proof}

Having analyzed Algorithm~\ref{alg:allgraphs}, we are ready to prove our main algorithmic result.
\begin{apptheorem}{\textbf{\ref{thm:approximationratio}}} (Restated)\emph{ 
  There is a polynomial time approximation algorithm for \TSP with
  performance guarantee $\frac{14\cdot( \sqrt{2}-1)}{12\cdot
    \sqrt{2}-13} < 1.461$.}
\end{apptheorem}

\begin{proof}
  By Lemma~\ref{lemma:2connTSP} and the discussion before
  Algorithm~\ref{alg:allgraphs}, we can restrict ourselves to
  $n$-vertex graphs that are $2$-vertex connected and have at most
  $2n-1$ edges.  The statement now follows by using
  Algorithm~\ref{alg:allgraphs} if $\OLP{G}$ is close to $n$ and
  otherwise by using Christofides' algorithm. 

  On the one hand, since Christofides' algorithm returns a solution
  with at most $n-1 + \OLP{G}/2$ edges (see~\cite{SW90} for an
analysis of Christofides' algorithm in terms of \OLP{G}), it has an approximation guarantee of at
  most
$$
\frac{n + \OLP{G}/2}{\OLP{G}}.
$$
On the other hand, by Lemma~\ref{lemma:circcost}, the approximation
guarantee of Algorithm~\ref{alg:allgraphs} is at most
$$
\frac{\frac{4}{3} n + \frac{2}{3} \left( 6(1-\sqrt{2})n + (4\sqrt{2} -3)\OLP{G} \right) }{\OLP{G}}.
$$
In particular, the approximation guarantee of
Algorithm~\ref{alg:allgraphs} for a graph $G$ with $\OLP{G}=n$ is $4/3 +
2/3\cdot (\sqrt{2}-1)^2 \approx 1.4477$ but deteriorates as \OLP{G}
increases. The approximation guarantee of Christofides' algorithm on
the other hand is getting better and better as \OLP{G} increases.
\begin{figure}[tb]
    \begin{center}
        \includegraphics[width=10cm]{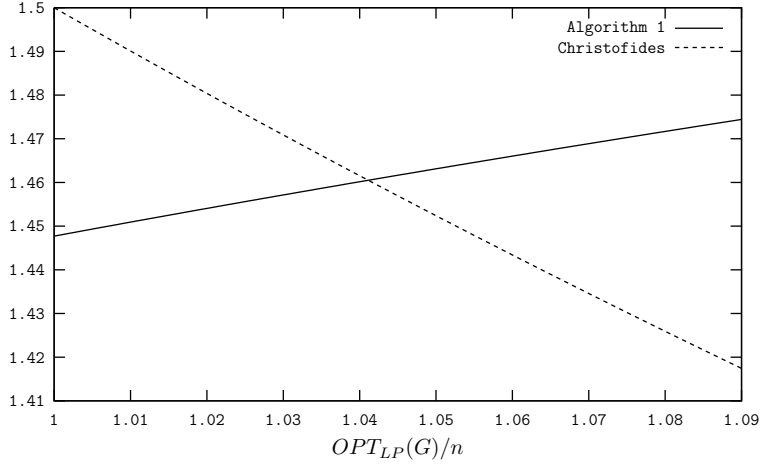}
    \end{center}
    \caption{The approximation ratios of
    Algorithm~\ref{alg:allgraphs} and Christofides' algorithm depending on
    the ratio $\OLP{G}/n$.}
    \label{fig:ratios}
\end{figure}
Comparing these two ratios, one gets that the worst case happens when
$\OLP{G} = \frac{24\sqrt{2}-26}{16\sqrt{2}-15} n$ (see Figure~\ref{fig:ratios}) and, by using simple
arithmetics, the approximation guarantee can be seen to be
$\frac{14(\sqrt{2}-1)}{12\cdot \sqrt{2}-13}$.
\end{proof}

\section{The Traveling Salesman Path Problem}\label{sec:tspp}
\label{sec:tspp}
In this section, we describe a sequence of generalizations and modifications of
the techniques that we previously presented for \TSP and conclude with improved
approximation algorithms for the traveling salesman path problem on graphic
metrics, \HPP.
\subsection{Using Held-Karp for Graph-TSPP}\label{sec:HKpath}
We can obtain a natural generalization of $\LP{G}$ to \HPP by distinguishing
whether the end vertices $s$ and $t$ are in the same set of
vertices. To this end, let $\Phi = \{S \subseteq V \mid \{s,t\}
\subseteq S \mbox{ or } S \cap \{s,t\} = \emptyset \}$.  Then the
relaxation can be written as
\begin{equation*}
\LP{G,s,t}: \qquad
\begin{aligned}[t]
  \min & \sum_{e\in E}  x_e \\[2mm]
  x(\delta(S)) & \geq  2, & \emptyset \neq S \subset V, S \in \Phi\\[2mm]
  x(\delta(S)) & \geq  1, & \emptyset \neq S \subset V, S \notin \Phi\\[2mm]
  x & \geq 0.
\end{aligned} 
\end{equation*}
We denote the optimum of this generalized linear program by
\OLP{G,s,t}.  It is not hard to see that $\OLP{G} = \OLP{G,s,s}$. 

The graph on the right-hand-side in Figure~\ref{fig:intgap} has a fractional
solution such that the
integrality gap of $\LP{G,s,t}$ is lower bounded by 1.5.

\begin{figure}[tb]
    \begin{center}
        \includegraphics[width=12cm]{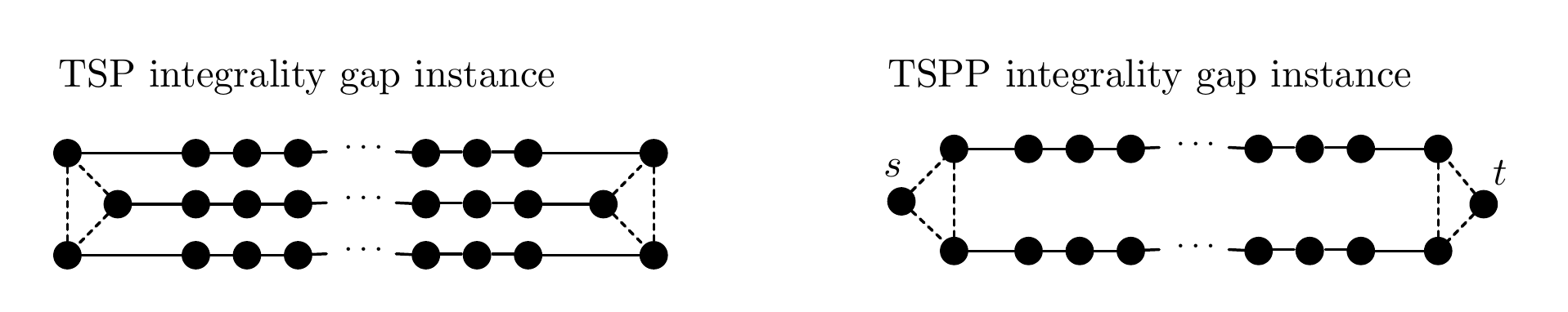}
    \end{center}
    \caption{Graphs for which the Held-Karp relaxation and the Held-Karp
    relaxation adapted to \HPP have an integrality gap tending to $4/3$ and $1.5$,
    respectively.}
    \label{fig:intgap}
\end{figure}

For a given graph $G=(V,E)$, let $G'=(V,E \cup \{e'\})$ be the graph obtained from
$G$ by inserting $e'=\{s,t\}$. Note that, given any solution $x$ to $\LP{G,s,t}$, we can
obtain a feasible solution to $\LP{G'}$ by adding 1 to $x_{e'}$. This way, for each of
the cuts where $S \notin \Phi$, we have $\delta(S) \ge 2$ and thus
$\OLP{G'} \le \OLP{G,s,t}+1$.
In the following, we will generalize our results for \TSP by using $\OLP{G'}-1$
as lower bound.

Similar to \TSP, we observe that the difficulty in approximating \HPP
lies in approximating those instances that are $2$-vertex connected.
The proof of this lemma can be found in Appendix~\ref{app:2connTSP}.
\begin{applemma}{\textbf{\ref{lemma:2connTSP}}} (Generalized) \emph{
Let $G$ be a graph and let $\mathcal{A}$ be an algorithm that, given a
$2$-vertex-connected subgraph $H$ of $G$ and $s,t \in V(H)$, returns a
\HPP{} solution to $(H,s,t)$ with cost at most $r \cdot
\OLP{H,s,t}$. Then there is an algorithm
$\mathcal{A}'$ that returns a \HPP{} solution to $(G,s,t)$ for any
$s,t\in V(G)$ with cost at most $r \cdot \OLP{G,s,t}$. Furthermore, the
running time of $\mathcal{A}'$ is a polynomial in the running time of
$\mathcal{A}$.}
\end{applemma}

\subsection{Generalization of the Approximation Framework to \HPP}

We generalize the framework to the problem \HPP. We obtain an
approximation ratio that depends on $\dist{s,t}$, the distance of $s$
and $t$. Therefore we can see the variant of Theorem~\ref{thm:main}
for \TSP as a special case where $s$ and $t$ have the distance 0.

\begin{apptheorem}{\textbf{\ref{thm:main}}} (Generalized) \emph{
  Given a $2$-vertex connected graph $G=(V,E)$ with a \MS{} $(R,P)$ and
  $s,t\in V$, there is a polynomial time algorithm that returns
  a spanning subgraph $H$ of $G$ with an Eulerian path between $s$ and $t$ with at most
  $\frac{4}{3}|E| - \frac{2}{3} |R| + \dist{s,t}/3$ edges.}
\end{apptheorem}
\begin{proof}
  A graph has an Eulerian path between $s$ and $t$ if and only if it
  is connected and the multigraph obtained by adding the edge $e'=\{s,t\}$ is
  a spanning Eulerian subgraph. Therefore, we basically want to apply
  (the original) Theorem~\ref{thm:main} and swap the degree of $s$ and $t$.

  To this end we create the graph $G'=(V,E')$ from $G$ by adding
  the edge $e'$ to $E$ if it is not already present in
  $G$. Then we apply Theorem~\ref{thm:main} to $G'$ with the removable
  pairing $(R, P)$ to obtain the spanning Eulerian subgraph $\tilde{G}$.
  
  If the Eulerian graph $\tilde{G}$ contains exactly one copy of $e'$,
  we simply remove it to obtain $H$. This case appears if and only if $e'$ was not chosen
  during the sampling, which happens with a probability of $2/3$. Note that the 
  2-edge-connectedness ensures that the removal does not disconnect $\tilde{G}$.
  
  Otherwise, with probability $1/3$, $\tilde{G}$ contains either two copies
  of $e'$ if $e' \notin R$ or none if $e' \in  R$. In either case we obtain $H$
  from $\tilde{G}$ by removing all copies
  of $e'$ and adding a shortest
  path of length exactly $\dist{s,t}$ to $\tilde{G}$.
  If $e' \in E$, we add a path with probability $1/3$ and apart from that we
  only remove edges; the claimed result follows immediately. 
  If $e' \notin E$, it is also not in $R$ and thus the path is
  added if and only if two edges are removed. Furthermore, with probability
  $2/3$, one edge is removed. Then the expected
  number of edges in $H$ is
\[
\frac{4}{3} (|E|+1) - \frac{2}{3} |R| + \frac{\dist{s,t}-2}{3} - 2/3 =
\frac{4}{3} |E| - \frac{2}{3} |R| + \frac{\dist{s,t}}{3}.
\]
 
Both the
removal of $e'$ and adding the shortest path swaps the parities
of $s$ and $t$, but of no other vertex.
\end{proof}

By using the generalized Theorem~\ref{thm:main} within the proof of
Lemma~\ref{lemma:costcirc}, we obtain immediately the following generalization.

\begin{applemma}{\textbf{\ref{lemma:costcirc}}} (Generalized) \emph{
  Given a $2$-vertex connected graph $G$, two vertices $s,t$ in $G$, and a depth first search tree
  $T$ of $G$, let $C^*$ be the minimum cost circulation to
  $C(G,T)$ of cost $c(C^*)$. Then there is a spanning multigraph $H$ of $G$ that
  has an Eulerian path between $s$ and $t$ with at
  most $\frac{4}{3}n + \frac{2}{3} c(C^*) - 2/3 + \dist{s,t}/3$ edges.}
\end{applemma}

\subsection{Approximation Algorithms for Graph-TSPP}
We are now equipped with the right tools to obtain algorithmic results for \HPP.
\begin{apptheorem}{\textbf{\ref{thm:approximationratiohpp}}} (Restated)\emph{
  For any $\varepsilon > 0$, there is a polynomial time approximation algorithm for
  \HPP with performance guarantee 
  $3-\sqrt{2} + \varepsilon < 1.586 + \varepsilon.$}

\emph{
If furthermore each block of the given graph is degree three bounded, there is a
polynomial time approximation algorithm for \HPP with performance
guarantee $1.5 + \varepsilon$, for any $\varepsilon>0$}.
\end{apptheorem}
\begin{proof}
By the generalized variant of Lemma~\ref{lemma:2connTSP}, it is sufficient to
show the theorem assuming that $G$ is 2-vertex connected.

If $G$ is degree three bounded, we apply Lemma~\ref{lem:boundeddeg} on $G$,
but use the generalized version of Lemma~\ref{lemma:costcirc} to obtain a
solution to \HPP that has at most $4n/3 - 2/3 + \dist{s,t}/3$ edges.
Additionally we may replace $\dist{s,t}$ by $n/2$, since in 2-vertex-connected graphs with more than two
vertices there are two vertex-disjoint paths between $s$ and $t$.

To obtain the claimed approximation ratio, we use the trivial lower bound $n-1$ of $\OLP{G,s,t}$.
For any $\varepsilon$, we determine a constant $n_0$ such that, for all $n \ge
n_0$, the approximation ratio is bounded from above by $1.5+\varepsilon$. If the graph has
fewer than $n_0$ vertices, we compute an optimal solution in constant time.

We continue with the case of general unweighted graphs.
As in the previous subsections, $e'=\{s,t\}$.
We apply Algorithm~\ref{alg:allgraphs} to obtain a
circulation $C'^*$ of $G'=(V,E \cup \{e'\})$ such that,
by Lemma~\ref{lemma:circcost},
$c(C'^*) \le 6(1-\sqrt{2})n + (4\sqrt{2} -3)\OLP{G'}$.
Using this circulation, we apply the generalized version of
Lemma~\ref{lemma:costcirc}. However, if $e' \notin E$ and it is used in the
solution (i.\,e., it was added as a shortest path), we have to replace $e'$ by a
shortest path between $s$ and $t$ in $G$.
This is equivalent to using $\dist{s,t}$ from $G$ instead of $G'$ in
Lemma~\ref{lemma:costcirc}. Therefore, in the following $\dist{s,t}$ always
refers to the distance in $G$ and
we obtain a solution to \HPP of at most
\begin{eqnarray*}
    &&\frac{4}{3}n + \frac{2}{3}(6(1-\sqrt{2})n + (4\sqrt{2} -3)\OLP{G'}) -
        \frac{2}{3} + \frac{\dist{s,t}}{3} \nonumber\\
    &=&(16/3 - 4\sqrt{2})n + \dist{s,t}/3 + (8 \sqrt{2}/3 - 2)(\OLP{G'}) - 2/3
\end{eqnarray*}
edges.

In the following, let $d=\dist{s,t}/n$ and $\zeta = (\OLP{G'}-1)/n$. Then, using
the lower bound $\OLP{G'}-1$ on $\OLP{G,s,t}$, the
approximation ratio achieved by our algorithm is at most
\begin{equation}\label{eqn:hppapprox}
\frac{16/3-4\sqrt{2}+d/3}{\zeta} + 8\sqrt{2}/3 - 2 + \epsilon_1,
\end{equation}
where $\zeta \ge 1 - 1/n$ and $\epsilon_1 = (8\sqrt{2}/3-8/3)/\OLP{G'}$.
In the following calculations, we omit $\epsilon_1$, since it decreases with the
input size. Similarly, we assume $\zeta \ge 1$. We will consider the deviation,
however, in the final result.

Since (\ref{eqn:hppapprox}) depends on $\zeta$, similar to the case of \TSP
we employ a second algorithm to obtain an upper bound independent of $\zeta$.

Let $\mathcal{A}$ be the following simple approximation algorithm for \HPP which
can be considered folklore.
First, $\mathcal{A}$ computes a spanning tree $T$ of cost $n-1$ in $G$. Then
$\mathcal{A}$ doubles all edges but those on the unique path between $s$ and $t$
in $T$.

The output of $\mathcal{A}$ is clearly a valid solution to \HPP and it computes
a solution of at most $2\cdot(n-1)-\dist{s,t}$ edges. Similar to
(\ref{eqn:hppapprox}), this results in an approximation ratio of at most
\begin{equation}\label{eqn:hppapproxtwo}
    (2-d)/\zeta.
\end{equation}

Note that for $\zeta=1$ and $d= \sqrt{2}-1$, disregarding $\varepsilon$,
(\ref{eqn:hppapproxtwo}) is the approximation ratio we are aiming for. Any
increase of $\zeta$ or $d$ can only improve this ratio. Therefore we may
restrict the analysis to values of $d$ in the range $[0,\sqrt{2}-1]$.

We will first analyze the approximation ratio depending on $d$ and determine
afterwards the value of $d$ where the minimum of the two
approximation ratios is maximized.

For any fixed $d$ within the considered range, (\ref{eqn:hppapprox}) is
monotonically increasing with respect to $\zeta$,
whereas (\ref{eqn:hppapproxtwo}) is monotonically decreasing. Since we are
interested in the minimum of the ratios, in the worst
case both ratios are equal. This happens when
\begin{equation}\label{eqn:equalhpp}
    \zeta = \frac{12\sqrt{2}-4d-10}{8\sqrt{2}-6}.
\end{equation}

We now replace $\zeta$
by (\ref{eqn:equalhpp}) in (\ref{eqn:hppapproxtwo}) to obtain the worst case
approximation ratio depending on $d$
\[
\frac{8\sqrt{2}-6-4d\sqrt{2}+3d}{6\sqrt{2}-2d-5}.
\]
Since this ratio can be seen to be monotonically increasing with respect to $d$ within the
considered range, the worst case appears when $d=\sqrt{2}-1$,
and thus we obtain as upper bound on the approximation ratio
\[
\frac{8\sqrt{2}-6-4(\sqrt{2}-1)\sqrt{2}+3(\sqrt{2}-1)}{6\sqrt{2}-2(\sqrt{2}-1)-5}
= \frac{15\sqrt{2}-17}{4\sqrt{2}-3} = 3 - \sqrt{2}.
\]

To conclude the proof, we still have to consider $\varepsilon_1$ and the case
where $\zeta<1$.
For any
$\varepsilon>0$, we determine an $n_0$ based on $\varepsilon_1$ and $\zeta$ similar to the degree bounded case and
solve \HPP on graphs with fewer than $n_0$ vertices exactly.
Altogether, we obtain an approximation ratio of
at most 
\[
3-\sqrt{2} + \varepsilon
\]
(see also Figure~\ref{fig:ratioshpp}).
\end{proof}
\begin{figure}[tb]
    \begin{center}
        \includegraphics[width=10cm]{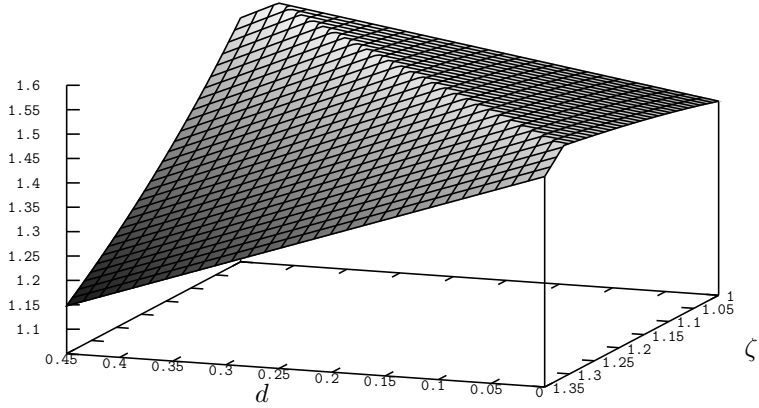}
    \end{center}
    \caption{The minimum of the approximation ratios (\ref{eqn:hppapprox}) and
    (\ref{eqn:hppapproxtwo}) depending on $d$ and $\zeta$.}
    \label{fig:ratioshpp}
\end{figure}
\section{Conclusions}
We have introduced a framework of removable pairings to find Eulerian
multigraphs.  This framework proved to be useful to obtain an
approximation algorithm for \TSP with an approximation ratio smaller
than $1.461$ and to obtain a tight upper bound on the integrality gap
of the Held-Karp relaxation for a restricted class of graphs that
contains degree three bounded and claw-free graphs. In particular, we
showed that in subcubic $2$-vertex-connected graphs we can always find a
solution to \TSP of at most $4n/3 - 2/3$ edges, which settles a
conjecture from \cite{BSSS11} affirmatively.

Our framework is not restricted to \TSP. With the same techniques
and a more detailed analysis, our result translates to the
traveling salesman path problem on graphic metrics with prespecified start and
end vertex.
In this way, one is guaranteed to obtain an approximation ratio smaller than $1.586$
and, for the degree three bounded case, the
approximation ratio gets arbitrarily close to $1.5$.

We note that the framework of removable pairings is straightforward
to generalize to general metrics, but the problem of finding a large
enough removable pairing in such graphs in order to improve on Christofides'
algorithm remains open.

\pagebreak
\appendix

\section{Example of Circulation Network}
\begin{figure}[hbtp]
\begin{center}
\includegraphics[width=12cm]{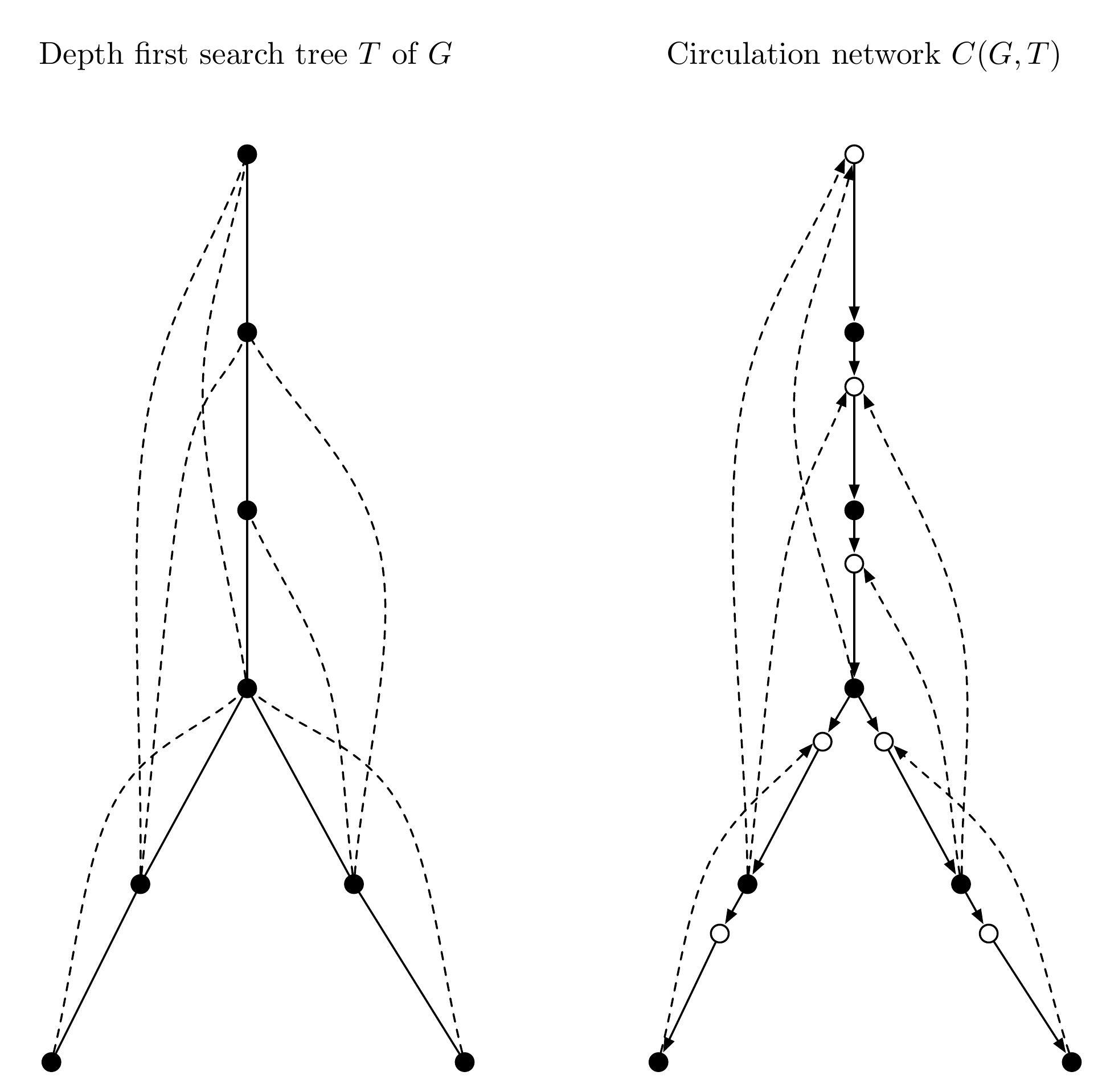}
\end{center}
\caption{The circulation network $C(G,T)$ of a graph $G$ with depth-first-tree $T$. In-vertices and out-vertices of the circulation network is depicted in white and black, respectively.}
\label{fig:circreplace_appendix}

\end{figure}

\section{Omitted Proofs}

\subsection{Proof of Lemma~\ref{lemma:2connTSP}}
\label{app:2connTSP}
We prove the more general lemma from Section~\ref{sec:tspp} that also applies to the traveling
salesman path problem.
\begin{applemma}{\textbf{\ref{lemma:2connTSP}}} (Restated) \emph{
Let $G$ be a graph and let $\mathcal{A}$ be an algorithm that, given a
$2$-vertex-connected subgraph $H$ of $G$ and $s,t \in V(H)$, returns a
\HPP{} solution to $(H,s,t)$ with cost at most $r \cdot
\OLP{H,s,t}$. Then there is an algorithm
$\mathcal{A}'$ that returns a \HPP{} solution to $(G,s,t)$ for any
$s,t\in V(G)$ with cost at most $r \cdot \OLP{G,s,t}$. Furthermore, the
running time of $\mathcal{A}'$ is a polynomial in the running time of
$\mathcal{A}$.}
\end{applemma}

\begin{proof}
  We define an $r$-approximation algorithm $\mathcal{A'}$ for $G$ as
  follows:
  \begin{enumerate}
  \item If $G$ is $2$-vertex connected then return the \HPP{} solution
    obtained by running $\mathcal{A}$ on $(G,s,t)$.
  \item Otherwise, let $v$ be a cut vertex whose removal results in
    components $C_1, C_2,\ldots,C_l$ with $l >1$. 
    Recursively run
      $\mathcal{A'}$ on the $l$ sub-instances $(G_1,s_1, t_1), \dots, (G_l, s_l, t_l)$ and return the union of the obtained solutions, where $G_i$  denotes the subgraph of $G$ induced by $C_i \cup \{v\}$, 
      $$s_i = \begin{cases} 
        s & \mbox{if }s\in C_i \\ 
        v & \mbox{otherwise}
      \end{cases}\qquad  \mbox{and} \qquad t_i = \begin{cases} 
        t & \mbox{if }t\in C_i \\ 
        v & \mbox{otherwise}
      \end{cases}.$$
      
  \end{enumerate}
  As a vertex is selected to be a cut vertex at most once,
  $\mathcal{A'}$ terminates in time bounded by a polynomial in the
  running time of $\mathcal{A}$.  It remains to verify that it returns
  a \HPP{} solution to $(G,s,t)$ with cost at most $r \cdot
  \OLP{G,s,t}$. We do so by induction on the depth of the
  recursion. In the base case no recursive calls are made so the
  solution is that returned by $\mathcal{A}$ which by assumption is a
  \HPP{} solution to $(G,s,t)$ with cost at most $r \cdot
  \OLP{G,s,t}$.

  Now consider the inductive step when a cut vertex $v$ of $G$ is
  selected whose removal results in components $C_1, C_2, \ldots, C_l$
  with $l > 1$.  Let $E_i$ be the multiset of edges of the obtained
  \HPP{} solution to $(G_i, s_i, t_i)$. With this notation the edge
  set returned by $\mathcal{A}'$ is $\bigcup_{i=1}^\ell E_i$ and we
  need to prove that 
  \begin{itemize}
  \item[(a)] it is a feasible \HPP{} solution to $(G, s,t)$, i.e, the
    edge set $\bigcup_{i=1}^\ell E_i \cup \{s,t\}$ forms a spanning
    Eulerian subgraph; and
  \item [(b)] $\sum_{i=1}^\ell |E_i| \leq r \cdot \OLP{G,s,t}.$
\end{itemize}

We start by proving (a). By the induction hypothesis, the edge set
$E_i \cup \{s_i, t_i\}$ forms a spanning Eulerian subgraph of $G_i$
and, consequently, $\bigcup_{i=1}^\ell \left(E_i \cup \{s_i,
  t_i\}\right)$ forms a spanning Eulerian subgraph of $G$.  That
$\bigcup_{i=1}^\ell E_i \cup \{s, t\}$ is a spanning Eulerian subgraph
of $G$ now follows from that the endpoints of $\{s_1,t_1\}, \{s_2,
t_2\}, \ldots, \{s_\ell, t_\ell\}$ 
can be partitioned so that one is $s$, one is $t$ and the
remaining $2(\ell-1)$ endpoints are $v$ ( possibly not different from
$s$ and $t$).

We proceed by proving (b).  By the induction hypothesis, 
$ \sum_{i=1}^\ell |E_i| \leq r \cdot \sum_{i=1}^\ell \OLP{G_i, s_i,
  t_i} $ and it is thus sufficient to prove $\sum_{i=1}^\ell \OLP{G_i,
  s_i, t_i} \leq \OLP{G,s,t}$.  To this end, Let $x$ be an optimal
solution to $\LP{G,s,t}$ and let $x^i$ denote its restriction to the
subgraph $G_i$ with start vertex $s_i$ and end vertex $t_i$.  By the
definition of $G_i, s_i, t_i$ and the fact that $v$ is a cut vertex,
it is easy to see that each constraint in $\LP{G_i, s_i, t_i}$ has an
identical constraint in $\LP{G,s,t}$. Therefore, $x^i$ corresponds to
a solution to \LP{G_i, s_i, t_i} and hence $\OLP{G,s,t} \geq
\sum_{i=1}^\ell \OLP{G_i,s_i,t_i}, $ which completes the inductive
step and the proof of the lemma.
\end{proof}

\subsection{Maximum Profit to Size Ratio}
\label{sec:maxsizratio}
We verify that $\max_{0\leq T\leq 1} \frac{1-T}{2-T} T$ is obtained
when $T= 2-\sqrt{2}$.  Let $f(T) = \frac{1-T}{2-T} T = \frac{T}{2-T} -
\frac{T^2}{2-T}$ and consider its first derivative
$$
\frac{d}{dT}f(T) = \frac{1}{2-T} + \frac{T}{(2-T)^2}  - \left(\frac{2T}{2-T} + \frac{T^2}{(2-T)^2}\right) = \frac{1-2T}{2-T} + \frac{T-T^2}{(2-T)^2}.
$$
From this it follows that $\frac{d}{dT}f(T) = 0$ when
$$
(1-2T)(2-T) + T-T^2 = 0 \Leftrightarrow T^2 - 4T + 2 = 0 \Leftrightarrow T = 2 \pm \sqrt{2}.
$$
It is now easy to verify that the unique maximum of $f(T)$ for $0\leq
T\leq 1$ is obtained when $T= 2-\sqrt{2}$.

\begin{thebibliography}{10}

\bibitem{Arora98}
Sanjeev Arora.
\newblock Polynomial time approximation schemes for {E}uclidean traveling
  salesman and other geometric problems.
\newblock {\em Journal of the ACM}, 45:753--782, 1998.

\bibitem{Barahona04}
Francisco Barahona.
\newblock Fractional packing of {T}-joins.
\newblock {\em SIAM Journal on Discrete Mathematics}, 17:661--669, 2004.

\bibitem{BSSS11}
Sylvia Boyd, Rene Sitters, Suzanne van~der Ster, and Leen Stougie.
\newblock {TSP} on cubic and subcubic graphs.
\newblock In {\em Proc.~of the 15th Conference on Integer Programming and
  Combinatorial Optimization (IPCO 2011)}, 2011.
\newblock To appear.

\bibitem{Chr76}
Nicos Christofides.
\newblock Worst-case analysis of a new heuristic for the travelling salesman
  problem.
\newblock Technical Report 388, Graduate School of Industrial Administration,
  Carnegie-Mellon University, 1976.

\bibitem{CFN85}
Gérard Cornuéjols, Jean Fonlupt, and Denis Naddef.
\newblock The traveling salesman problem on a graph and some related integer
  polyhedra.
\newblock {\em Mathematical Programming}, 33:1--27, 1985.

\bibitem{Edmonds1965b}
Jack Edmonds.
\newblock Maximum matching and a polyhedron with $0,1$ vertices.
\newblock {\em Journal of Research of the National Bureau of Standards},
  69:125--130, 1965.

\bibitem{FJJ89}
Greg~N. Frederickson and Joseph Ja'ja'.
\newblock On the relationship between the biconnectivity augmentation and
  travelling salesman problems.
\newblock {\em Theoretical Computer Science}, 19(2):189 -- 201, 1982.

\bibitem{GLS05}
David Gamarnik, Moshe Lewenstein, and Maxim Sviridenko.
\newblock An improved upper bound for the {TSP} in cubic 3-edge-connected
  graphs.
\newblock {\em Operations Research Letters}, 33(5):467--474, 2005.

\bibitem{GSS11}
Shayan~Oveis Gharan, Amin Saberi, and Mohit Singh.
\newblock A randomized rounding approach to the traveling salesman problem.
\newblock preprint, 2011.

\bibitem{GB90}
Michel~X. Goemans and Dimitris~J. Bertsimas.
\newblock On the parsimonious property of connectivity problems.
\newblock In {\em Proceedings of the 1st Annual ACM-SIAM Symposium on Discrete
  Algorithms (SODA~1990)}, pages 388--396, 1990.

\bibitem{Goe95}
Michel~X. Goemans.
\newblock Worst-case comparison of valid inequalities for the {TSP}.
\newblock {\em Mathematics and Statistics}, 69:335--349, 1995.

\bibitem{GKP95}
Michelangelo Grigni, Elias Koutsoupias, and Christos~H. Papadimitriou.
\newblock An approximation scheme for planar graph {TSP}.
\newblock In {\em Proc.~of the 36th Annual Symposium on Foundations of Computer
  Science (FOCS~1995)}, pages 640--645, 1995.

\bibitem{GLS1988}
Martin Gr{\"o}tschel, L{\'a}szlo Lov{\'a}sz, and Alexander Schrijver.
\newblock {\em {Geometric Algorithms and Combinatorial Optimization}}, volume~2
  of {\em Algorithms and Combinatorics}.
\newblock Springer, 1988.

\bibitem{HK70}
Michael Held and Richard~M. Karp.
\newblock The traveling-salesman problem and minimum spanning trees.
\newblock {\em Operations Research}, 18:1138--1162, 1970.

\bibitem{Hoo91}
J.~A. Hoogeveen.
\newblock Analysis of {C}hristofides' heuristic: some paths are more difficult
  than cycles.
\newblock {\em Operations Research Letters}, 10(5):291--295, 1991.

\bibitem{KKN01}
Atsushi Kaneko, Alexander Kelmans, and Tsuyoshi Nishimura.
\newblock On packing 3-vertex paths in a graph.
\newblock {\em Journal of Graph Theory}, 36(4):175--197, 2001.

\bibitem{Mitch99}
Joseph S.~B. Mitchell.
\newblock Guillotine subdivisions approximate polygonal subdivisions: A simple
  polynomial-time approximation scheme for geometric {TSP}, k-{MST}, and
  related problems.
\newblock {\em SIAM Journal on Computing}, 28:1298--1309, March 1999.

\bibitem{MMP90}
C.~L. Monma, B.~S. Munson, and W.~R. Pulleyblank.
\newblock Minimum-weight two-connected spanning networks.
\newblock {\em Mathematical Programming}, 46:153--171, 1990.

\bibitem{NP81}
Denis Naddef and Wiliam~R. Pulleyblank.
\newblock Matchings in regular graphs.
\newblock {\em Discrete Mathematics}, 34(3):283--291, 1981.

\bibitem{PV06}
Christos~H. Papadimitriou and Santosh Vempala.
\newblock On the approximability of the traveling salesman problem.
\newblock {\em Combinatorica}, 26(1):101--120, 2006.

\bibitem{Sch03}
Alexander Schrijver.
\newblock {\em Combinatorial Optimization}.
\newblock Springer, 2003.

\bibitem{SW90}
David~B. Shmoys and David~P. Williamson.
\newblock Analyzing the {H}eld-{K}arp {TSP} bound: a monotonicity property with
  application.
\newblock {\em Information Processing Letters}, 35(6):281--285, 1990.

\bibitem{Wol80}
Laurence~A. Wolsey.
\newblock Heuristic analysis, linear programming and branch and bound.
\newblock In {\em Combinatorial Optimization II}, volume~13 of {\em
  Mathematical Programming Studies}, pages 121--134. Springer, 1980.

\end{thebibliography}
\end{document}